\documentclass[11pt]{article}
\usepackage[margin=1.0in]{geometry}
\usepackage{amsmath}
\usepackage{xfrac}
\usepackage{graphicx}
\usepackage{units}
\usepackage{verbatim}
\usepackage{listings}
\usepackage{algorithm}
\usepackage{color}
\usepackage{algorithmicx}
\usepackage[noend]{algpseudocode}
\usepackage{comment}
\usepackage{verbatimbox}
\usepackage{booktabs}
\usepackage{url}

\usepackage{amssymb}

\usepackage[parfill]{parskip}

\usepackage{fontspec}
\usepackage{fancyvrb}
\DefineVerbatimEnvironment{leancode}{Verbatim}
  {fontsize=\scriptsize, numbers=left, numbersep=6pt}
\usepackage{listings}
\let\origref\ref
\def\ref#1{\textbf{\origref{#1}}}

\usepackage[inkscapelatex=false]{svg}
\usepackage[backend=bibtex]{biblatex}

\usepackage{quotes}

\everymath{\displaystyle}
\usepackage{amsthm}
\theoremstyle{definition}
\newtheorem{thm}{Theorem}[section]
\newtheorem{lemma}{Lemma}[section]

\usepackage{unicode-math}

\newcommand{\N}{\mathbb{N}}
\newcommand{\Q}{\mathbb{Q}}

\newcommand{\Z}{\mathbb{Z}}

\usepackage{authblk}

\title{Machine-Verifying Toom-Cook Multiplication with Integer Evaluation Points}
\author{Srihari Nanniyur}
\author{Siddhartha Jayanti}
\affil{Dartmouth College}
\date{}

\begin{document}
\maketitle

\begin{abstract}
We present a machine-verified proof of the correctness of Toom-Cook multiplication with generalized integer evaluation points. Toom-Cook is a class of fast multiplication algorithms parameterized by a triple $(k_x, k_y, \vec v)$ consisting of two positive integer \emph{split sizes} $k_x, k_y$ and a vector $\vec v$ of distinct \emph{evaluation points}. As part of our proof, we verify that for any selection of $k_x+k_y-1$ distinct integer evaluation points, we can obtain a formula for a base case threshold $\theta(k_x, k_y, \vec v)$ sufficient to ensure termination regardless of the values of the operands. The threshold formula, which we derive by obtaining upper bounds on the subproblem sizes produced by the Toom-Cook recurrence, does not depend on the operands; it depends only on $k_x$, $k_y$, $\vec v$, and the base $b$ in which we operate.

We write the proof in Lean 4, making use of the Mathlib library. We formalize the algorithm, our base case threshold formula, and our key lemma statements in Lean. We then use the AI theorem prover Aristotle to assist in completing the machine verification of the algorithm's correctness. This proof, through its synthesis of human input and AI assistance, demonstrates the continuing capability of AI to automate the machine verification process.
\end{abstract}

\section{Introduction}

Large integer multiplication has applications across mathematics, computing, and physics.
While the standard \emph{long multiplication} (a.k.a. \emph{grade-school multiplication}) algorithm takes $O(n^2)$ time to multiply two operands of $n$ digits each, there are several \emph{fast multiplication} algorithms that achieve an asymptotic speed-up, including Karatsuba ($O(n^{\log_2 3})$) \cite{karatsuba-ofman}, Toom-Cook ($O(n^{1+\varepsilon})$) \cite{toom, cook}, Schönhage-Strassen ($O(n \log n \log\log n)$) \cite{schonhage-strassen}, and Harvey-van der Hoeven ($O(n \log n)$) \cite{h-vdh-21}.

Toom-Cook, the focus of this paper, is actually a \emph{class} of algorithms parameterized by a triple $(k_x, k_y, \vec v)$ consisting of two positive integer \emph{split sizes} $k_x, k_y$ and a vector $\vec v$ of distinct \emph{evaluation points}.
In particular, to multiply $x$ and $y$, the algorithm uses a divide-and-conquer approach that splits $x$ into $k_x$ pieces, splits $y$ into $k_y$ pieces, and constructs two polynomials $p$ and $q$ whose coefficients are the pieces of $x$ and $y$ respectively.
It evaluates the polynomials at each of the points in the vector $\vec v = \langle v_1,\cdots, v_{k_x+k_y-1} \rangle$, computes pointwise products to determine the pointwise evaluation of $p \cdot q$, and uses polynomial interpolation to compute the final result.

In the common case where $k_x = k_y = k$, Toom-Cook achieves a time complexity of $O(n^{\log_k (2k-1)})$; the exponent can get arbitrarily close to 1 at the expense of a fast-growing constant term hidden by the big-$O$ \cite{crandall-pomerance}.
Toom-3, i.e., Toom-Cook where $k_x=k_y=3$, is widely used for fast multiplication of mid-to-large sized integers \cite{gmp}.
Proofs of correctness have been written for subclasses of the Toom-Cook family: for example, with $\vec{v} = \langle 0, 1,\cdots,2k-2 \rangle$ \cite{cook}. 
Part of such a proof is the derivation of a \emph{base case threshold}, $\theta$, which specifies the instance size---i.e., number of digits---at which to switch from the recursive case to the base case.

In this work, we derive a function $\theta(k_x, k_y, \vec v)$ that computes base case thresholds for arbitrary split sizes $k_x, k_y \geq 2$ and point vector $\vec v \in \Z^{k_x + k_y - 1}$, where the values in $\vec v$ are distinct.
We thereby---to our knowledge---provide the first machine-verified formalization of Toom-Cook that accounts for generalized integer evaluation points. (There are some other aspects that we do not generalize; see Section \ref{sec:algorithm}.) 
We then machine-verify an abstract mathematical specification of the algorithm in Lean.
Our proof combines manual derivation and analysis with the use of AI tools, including \emph{Aristotle} \cite{achim} for proving lemma goals in Lean and GPT-5 to aid in producing formalization strategies.

\subsection{Multiplication algorithms}
The standard \textit{long multiplication} (a.k.a. \textit{grade-school multiplication}) algorithm for integers has been known since antiquity.
The quadratic running time of this algorithm was conjectured to be asymptotically optimal by Kolmogorov \cite{karatsuba-retrospective}.
In 1962, Karatsuba \cite{karatsuba-ofman} refuted this conjecture by devising a divide-and-conquer algorithm that reduces multiplication of two $n$-digit operands to three multiplications with operands that have roughly $\sfrac{n}{2}$ digits, thus achieving a running time of $O\left(n^{\log_2 3}\right)$.

Toom \cite{toom} and Cook \cite{cook} vastly generalized Karatsuba's approach, showing that by dividing each problem into more subproblems of smaller size, multiplication can be performed in $O(n^{1+\varepsilon})$ time for arbitrarily small $\varepsilon > 0$.
The Toom-Cook class of algorithms is the primary subject of this paper; we detail it in Section \ref{sec:algorithm}.

Sch\"{o}nhage and Strassen showed that multiplication can be performed in $O(n \log n \log \log n)$ time using the Fast Fourier Transform (FFT) \cite{schonhage-strassen}.
They conjectured that $O(n \log n)$ would be achievable and optimal.
A recent series of papers \cite{furer, h-vdh-lecerf-poly, h-vdh-lecerf-2014, h-vdh-lecerf-2018} in this domain made steady progress on this conjecture.
This line of work culminated in the celebrated algorithm of Harvey and van der Hoeven \cite{h-vdh-21}, which uses more complicated FFT-based techniques to achieve $O(n \log n)$ time complexity. 
However, this algorithm is galactic, i.e., it is only efficient at problem sizes too large for any practical context \cite{h-vdh-21}. 
The question of asymptotic optimality remains open.
In fact, we are not aware of any non-trivial unconditional general lower bounds on the asymptotic time complexity of integer multiplication.

\subsection{Lean, Aristotle, and AI-assisted theorem proving} \label{sec:lean}

Our Toom-Cook proof is specified in Lean, a language for formalizing and machine-verifying mathematical proofs. Given a Lean formalization of a proof, the Lean kernel is able to check that proof's correctness. Lean provides a large number of \textit{tactics}, or commands specialized for solving particular types of proof goals. We also make extensive use of Mathlib, a library of mathematical tools and lemmas implemented in Lean.

Aristotle, an AI tool developed by Harmonic, is capable of generating valid Lean proofs by itself. Given a piece of Lean code containing unfilled lemma headers, Aristotle will attempt to fill in the proofs of those lemmas, making no changes to the rest of the code. Harmonic's October 2025 white paper on Aristotle \cite{achim} provides more information on the tool's internals. A significant portion of our proof code is Aristotle-generated. In accordance with the spirit of this project, we also made judicious use of GPT-5 to find Mathlib lemmas and help structure our code.

\subsection{Related work}
To our knowledge, we are the first to machine-verify Toom-Cook at the level of generality treated here.
We are also not aware of previous work that derives an explicit base case threshold formula for generalized $k_x$, $k_y$, and $\vec v$ (where $\vec v \in \mathbb{Z}^{k_x+k_y-1}$ and the points in $\vec v$ are distinct).
In this section, we outline other previous work on the Toom-Cook algorithm, including on optimizing the algorithm via split construction and evaluation point selection.

Cook \cite{cook}, building on the work of Toom \cite{toom}, provided the first rigorous specification of a subset of the Toom-Cook class of algorithms. Our own reference point is the description given in Brent and Zimmermann's \textit{Modern Computer Arithmetic} \cite{brent-zimm}. Zanoni \cite{zanoni} provides an analysis of the case where $k_x \neq k_y$, in particular of Toom--2.5 (that is, with splitting parameters 3 and 2). Zanoni designs an iterative method to efficiently apply Toom-2.5 to operands heavily unbalanced in size.

Bodrato \cite{bodrato}, both on his own and in collaboration with Zanoni \cite{bodrato-zanoni}, works towards developing heuristics for optimizing evaluation-point selection. Bodrato and Zanoni \cite{bodrato-zanoni} provide strategies for finding the most efficient sequences of operations to invert the matrices generated by these evaluation points. Larasati, Awaludin, Ji, and Kim \cite{larasati}, building on Bodrato's work, give a quantum circuit for Toom--3 multiplication. They minimize the number of division operations performed in the interpolation step, which lowers the Toffoli count of the circuit.

For the purposes of machine-verifying an abstract mathematical specification of the algorithm, we allow the evaluation points to be arbitrary fixed finite integers.
However, recent literature on Toom-Cook has not focused on the theoretical implications of this case, perhaps because selecting points of small magnitude works in practice.

\subsection{Challenges and our approach}

While the Toom-Cook family of algorithms is well understood, expositions on the general case have omitted the derivation of a base case threshold sufficient to ensure termination for generalized $\vec v$.
In the case of generalized $\vec v$, this value cannot be chosen as an absolute constant.
This is because the threshold required to ensure termination actually depends on the splitting parameters ($k_x$, $k_y$) and evaluation points ($\vec v$).
Large points lead to large pointwise evaluations, thus requiring a more careful analysis of the progression of subproblem sizes.

Writing machine-verifiable proofs is time-consuming, since it requires spelling out low-level details.
In the case of recursive algorithms like Toom-Cook, one of the most time-consuming steps in machine verification is proving termination.
The challenge lies in stating explicit bounds on the sizes of problems and resultant subproblems, and also in showing that the sizes are actually strictly decreasing between successive recursion depths above the base case.
We adopt an AI-driven approach to this task, taking advantage of the existence of an AI proof assistant that can push user-specified lemma goals through to completion.

Our first main contribution is the derivation of an explicit base case threshold formula $\theta(k_x, k_y, \vec v)$, such that Toom-Cook-($k_x$, $k_y$, $\vec v$) can use the recursive case until the \emph{problem size}---the digit count of the operand with larger absolute value---no longer exceeds $\theta$. The algorithm switches to the base case when the problem size is $\leq \theta(k_x, k_y, \vec v)$. (Note that the base $b$ in which we work is also implicit in the formula; all operations and representations in the algorithm take place in base $b$. Thus, we do not explicitly notate $b$ in the parameter list.)

The derivation of $\theta$ requires an explicit formulation of the problem size as a function of the digit counts of all the components that contribute to it.
We separate out those components that depend on the operands from those that do not; the result is a clean, immediately usable formula for a valid base-case threshold in terms of those components that do \textit{not} depend on the operands.

\subsection{Our contributions}
We make the following contributions in this paper.
\begin{itemize}
    \item \textbf{Base-case threshold formula for arbitrary distinct integer evaluation points.} 
    We derive a base case threshold formula $\theta(k_x, k_y, \vec{v})$ for Toom-Cook-$(k_x, k_y, \vec{v})$.
    Notably, the value of this threshold depends only on these three parameters and on the fixed base $b$ in which we operate; it does not depend on the operands $x$ and $y$.

    \item \textbf{Machine verification of Toom-Cook.} 
    We formulate a machine-verified proof, specified in Lean, of the class of Toom-Cook algorithms that use integer evaluation points, including cases where $k_x \neq k_y$.
    Notably, we make use of the AI theorem prover Aristotle to aid the proof process.
\end{itemize}

\section{Toom-Cook multiplication} \label{sec:algorithm}

We give a brief overview of the algorithm, building on the descriptions given by Brent and Zimmermann in \textit{Modern Computer Arithmetic} \cite{brent-zimm} and by Zanoni in his paper on iterative Toom-Cook methods \cite{zanoni}.

We work in a fixed base $b \in \mathbb{N}$ and assume that $b \geq 2$. All digit counts taken from this point forward are in base $b$. For the remainder of this paper, we work entirely in base $b$.

Let $\delta(a)$ denote the base-$b$ digit count of an operand $a \in \Z$. The operands to our multiplication are $x, y \in \mathbb{N}$. We define splitting parameters $k_x, k_y \in \mathbb{N}$, where $k_x \geq 2$ and $k_y \geq 2$. We define $P_{x, y} = \delta(\max\{x, y\})$ as the \textit{problem size}; or, to be more precise, the digit count in base $b$ of the operand with the greater magnitude. (For the extension to integer operands, we define $P_{x, y} = \delta(\max\{|x|, |y|\})$ for $x, y \in \Z$.) Finally, let $\vec v \in \mathbb{Z}^{k_x+k_y-1}$ be the vector of $k_x + k_y -1$ distinct integer evaluation points.

Let $B = b^i$ for some $i \in \mathbb{N}$ chosen such that $x < B^{k_x}$ and $y < B^{k_y}$. It is important to choose $i$ and $B$ such that the relative sizes of the chunks resulting from splitting an operand are kept as balanced as possible.

Our specification uses the efficient choice $i = \max\{\lceil \delta(x)/k_x \rceil\ , \lceil \delta(y) / k_y \rceil\}$. We do not further generalize the selection of $i$ and $B$, as our aim is to show that there exists at least one such selection for which we can derive a base case threshold satisfying termination and correctness for generalized $\vec v$. Our threshold derivation (Section \ref{sec:derivation}) is carried out with respect to this choice of $i$, although it is also correct for any $i$ that produces chunks of digit length at most $\lceil P_{x, y} / \min\{k_x, k_y\} \rceil$ in base $b$. 
Notably, Cook's thesis \cite{cook} specifies that $i$ (which he terms $q$) should satisfy a similar chunk-size bound.

We define the polynomials $p, q \in \Z[u]$ where $x = p(B)$ and $y = q(B)$:
\[
p(u) = \sum_{m=0}^{k_x-1} (\lfloor (x / B^m) \rfloor  \bmod B)u^m
\]
\[
q(u) = \sum_{m=0}^{k_y-1} (\lfloor (y / B^m) \rfloor \bmod B)u^m
\]

For each entry $v_j$ of $\vec v$, we recursively multiply $p(v_j)q(v_j)$ for sufficiently large $P_{x, y}$ and store the evaluated results in the vector $\vec w \in \mathbb{Z}^{k_x+k_y-1}$.

We take the Vandermonde matrix $V$ of $\vec v$. Since the values in $\vec v$ are distinct, this Vandermonde matrix will always be invertible provided the inversion and interpolation are performed over $\Q$. We multiply $V^{-1}\vec w$ to get the vector $\vec r$, which consists of the $k_x + k_y - 1$ coefficients of $p \cdot q$.

Let $h$ be the polynomial of degree at most $k_x+k_y-2$ whose coefficients are the values contained in $\vec r$. The construction of the Vandermonde interpolation step ensures that $h = p \cdot q$. We evaluate $h(B) = (p \cdot q)(B) = p(B) \cdot q(B) = x \cdot y$ to get the product of $x$ and $y$.

The procedure as specified above is for operands in $\mathbb{N}$. However, it is easy to extend it to operands in $\mathbb{Z}$ by multiplying those operands' absolute values and then determining the result's sign based on their original signs. This is the approach we have taken.

\paragraph{The infinity point.}
It is worth noting that in practice, a slightly expanded definition of $\vec v$ is often used, where $\infty$ is a point in $\vec v$. For the meaning of the evaluation at the infinity point, see Section \ref{sec:infty}. Our machine proof does \emph{not} account for the infinity point; but as we will see, our threshold derivation can be adapted to it easily enough.

\section{Base case threshold formula}
In practical implementations of Toom-Cook, the evaluation points are small enough (and the base-case cutoff high enough) that termination is not in question. However, if $\vec v$ is allowed to contain arbitrarily large integer points, then it becomes nontrivial to derive a closed form for a base case threshold sufficient to always ensure termination.

\subsection{Derivation} \label{sec:derivation}
We work entirely in base $b \in \N$, where $b \geq 2$. Recall that $\delta(a)$ denotes the number of digits (in base $b$) of an integer $a$. Recall as well that we defined $P_{x, y} = \delta(\max\{|x|, |y|\})$; that is, $P_{x, y}$ is the initial "problem size" with respect to the two operands  $x, y$. We will use $P$ as shorthand for $P_{x, y}$ from here onwards.

Given a pair of operands $x, y \in \mathbb{Z}$, it suffices to demonstrate this derivation for $|x|, |y| \in \mathbb{N}$, as we can multiply the absolute values of the operands and then determine the result's sign afterwards.

Let $v_{\max}$ be the absolute value (not the raw value) of the entry with the highest magnitude in $\vec v$. That is, we have that

\[
    v_{\max} = \max_j |v_j|
\]

Take operand $x$, and consider the polynomial $p \in \Z[u]$ obtained by splitting $|x|$ $k_x$ ways. We aim to separate out whatever depends on $x$ from whatever does not depend on it. A natural way to proceed with this is to define an expression for the coefficient of maximum magnitude.
\[
    p_{\max} = \max_\ell p_\ell
\]

Since the split is performed on $\{|x|, |y|\}$, the polynomials' coefficients are all non-negative. Using the monotonicity of $\delta$ over $\N$, we can state that

\begin{equation*}
\begin{aligned}
    \delta\left( p(v_j) \right) &= \delta\left( p_{(k_x - 1)} v_j^{(k_x - 1)} + \cdots + p_0 v_j^0 \right)\\
    &= \delta \left( |p_{(k_x - 1)} v_j^{(k_x - 1)} + \cdots + p_0 v_j^0 | \right) \quad \text{as $\forall s \in \Z$, $\delta(|s|) = \delta(s)$}\\
    &\leq \delta\left( |p_{(k_x - 1)} v_j^{(k_x - 1)}| + \cdots + |p_0 v_j^0| \right) \quad \text{by the triangle inequality}\\
    &= \delta\left( p_{(k_x - 1)} |v_j^{(k_x - 1)}| + \cdots + p_0 |v_j^0| \right) \quad \text{by non-negativity of coeffs.}\\
    &\leq \delta\left( p_{(k_x - 1)} v_{\max}^{(k_x - 1)} + \cdots + p_0 v_{\max}^0 \right)\\
    &\leq \delta \left( \sum_{m = 0}^{k_x - 1} p_{\max} v_{\max}^m \right)\\
    &= \delta\left( p_{\max} \sum_{m = 0}^{k_x - 1} v_{\max}^m \right)
\end{aligned}
\end{equation*}

and similarly for $q$, if one defines $q_{\max}$ in analogous fashion. Now observe that $\forall s, t \in \N$, $\delta(st) \leq \delta(s) + \delta(t)$. Thus, if we slightly loosen the upper bounds, we can re-express them in terms of digit counts: 
\begin{align*}
& \delta(p(v_j)) \leq \delta(p_{\max}) + \delta \left( \sum_{m = 0}^{k_x - 1} v_{\max}^m \right)\\
&\delta(q(v_j)) \leq \delta(q_{\max}) + \delta \left( \sum_{m = 0}^{k_y - 1} v_{\max}^m \right).
\end{align*}

Let
\[
    C_x = \delta \left( \sum_{m = 0}^{k_x - 1} v_{\max}^m \right)
\]
and
\[
    C_y = \delta \left( \sum_{m = 0}^{k_y - 1} v_{\max}^m \right).
\]

Under the splitting scheme selected earlier, $i = \max\{\lceil \delta(x)/k_x \rceil\ , \lceil \delta(y) / k_y \rceil\}$. Since coefficients are strictly less than $B = b^i$, we see that $\delta(p_{\max}), \delta(q_{\max}) \leq i \leq \lceil P / k_{\min} \rceil$, where $k_{\min} = \min\{k_x, k_y\}$. Thus, if we set $C = \max\{C_x, C_y\} + 1$ (the 1 accounting for the slack from the ceiling function), then we may calculate an upper bound on any given subproblem size $P'$ in terms of $P$. Specifically, let
\[
    P' = \delta\left( \max\{ |p(v_j)|, |q(v_j)| \} \right).
\]
Then
\begin{align*}
    P' &\leq i + \max\{C_x, C_y\}\\
    &\leq \lceil P / k_{\min} \rceil + \max\{C_x, C_y\}\\
    &\leq P/k_{\min} + C \quad \text{as $C = \max\{C_x, C_y\} + 1$}.
\end{align*}

Now, we solve for a value $\theta$ satisfying
\[
    \theta = \theta/k_{\min} + C
\]
from which one sees that
\[
    \theta\left( 1 - 1/k_{\min} \right) = C
\]
and
\begin{align*}
    \theta &= \frac{C}{1 - 1/k_{\min}}\\
    &= \frac{C}{\frac{k_{\min} - 1}{k_{\min}}}\\
    &= \frac{C k_{\min}}{k_{\min} - 1}
\end{align*}
and it now suffices to prove the following strict-decrease property.

\begin{thm} \label{thm1}
    If $P > \theta$, then $P' < P$.
\end{thm}
\begin{proof}
    We have
    \[
        \theta = \frac{Ck_{\min}}{k_{\min} - 1}
    \]
    and we showed earlier that
    \[
        P' \leq P/k_{\min} + C
    \]
    from which, recalling that $P > \theta$ and $k_{\min} \geq 2$, we obtain
    \begin{align*}
        P' &\leq \frac{P}{k_{\min}} + C\\
        &= \frac{P}{k_{\min}} + \theta\frac{k_{\min} - 1}{k_{\min}}\\
        &= \frac{P + \theta(k_{\min} - 1)}{k_{\min}}\\
        &= \frac{P - \theta}{k_{\min}} + \theta\\
        &< P - \theta + \theta \quad \text{as $P > \theta$ and $k_{\min} \geq 2$}\\
        &= P.
    \end{align*}
\end{proof}

\subsection{Extension to the infinity point} \label{sec:infty}
We now consider the problem of adapting our threshold derivation to a vector that includes the infinity point. This problem turns out to be solvable via a small update to our definition of $v_{\max}$. Note that our scope here is restricted to the $\theta$ derivation itself, i.e. we are still analyzing pointwise evaluation, not interpolation.

Consider $p \in \Z[u]$, the polynomial obtained from splitting $|x|$. The evaluation $p(\infty)$ is typically defined as the leading coefficient of $p$, and similarly for $q$. The precise definition---whether "leading" refers to the coefficient of $u^{k_x - 1}$ in a fixed-degree representation (which may be $0$) or whether it refers to the highest-degree nonzero coefficient---does not end up impacting the $\theta$ derivation. In fact, to prove Lemma \ref{lemma:infty}, a weaker condition suffices: we only need $p(\infty)$ to give \emph{any} single coefficient of $p$ for nonzero $p$, and $0$ if $p$ is the zero polynomial.

Now take $\vec v$, the vector of $k_x + k_y - 1$ distinct evaluation points. Suppose one of these points is $\infty$ and the rest are finite integers. We slightly amend our earlier definition of $v_{\max}$.
\[
    v_{\max} = \max_{\substack{v_j \in \vec v\\ v_j \neq \infty}} |v_j|.
\]

\begin{lemma} \label{lemma:infty}
    In the case with the infinity point included, $\delta(p(\infty)) \leq \delta(p(v_{\max}))$ and $\delta(q(\infty)) \leq \delta(q(v_{\max}))$.
\end{lemma}
\begin{proof}
    Recall that $k_x \geq 2$ and $k_y \geq 2$, or else we would be in a degenerate case. Therefore $\vec v$ contains at least 3 distinct points. One of these points is $\infty$, leaving at least two slots left. Since the points are distinct, there must exist at least one nonzero finite integer point in $\vec v$, as there cannot be more than one zero point. We see immediately that $v_{\max} \geq 1$, as $v_{\max}$ is the max over the \emph{absolute values} of finite integer points in $\vec v$.

    We may write $p(u)$ as
    \[
        p(u) = p_{(k_x - 1)} u^{(k_x - 1)} + \cdots + p_0 u^0
    \]
    and now recall a property of the evaluation at $\infty$: $p(\infty)$ gives one of the coefficients of $p$ (and $0$ if $p = 0$). Since we are dealing with the operands' absolute values here, all the coefficients are non-negative. So we have
    \[
        \delta(p(\infty)) \leq \delta(p_{(k_x - 1)} + \cdots + p_0) = \delta(p(1)) \leq \delta(p(v_{\max}))
    \]
    and we see that $\delta(p(\infty)) \leq \delta(p(v_{\max}))$. The same argument may be deployed symmetrically for $q$:
    \[
        \delta(q(\infty)) \leq 
        \delta(q_{(k_y - 1)} + \cdots + q_0) = \delta(q(1)) \leq \delta(q(v_{\max})).
    \]
\end{proof}

Thus, under our revised definition of $v_{\max}$, $\delta(p(v_{\max}))$ and $\delta(q(v_{\max}))$ still suffice as upper bounds on the digit counts of the polynomial evaluations, and the rest of the $\theta$ derivation goes through smoothly.

\section{Lean proof structure}
We provide an overview of the structure of the Lean proof. Note that our specification is an \emph{abstract mathematical description} of Toom-Cook. We do not aim to verify a particular implementation. Note also that our machine proof does not account for the infinity point.

The algorithm definition and high-level lemma goals were specified by us, and we used Aristotle's assistance to complete the proofs of those goals and of the intermediate results they entailed.
We also used GPT-5 to find lemmas in Mathlib and to help choose proof strategies for certain goals.
The proof comprises three key sections: the specification of the algorithm, the proof that the algorithm terminates, and the proof that the result returned by it is correct.

The proof is located in a GitHub repository \cite{repo} that also contains an Aristotle-assisted correctness proof of Karatsuba's algorithm. We have made substantial revisions to the proof since its first version; the major changes are recorded in the repository. The current version compiles under Lean v4.33.1 and Mathlib tag v4.33.1.

\subsection{Specification}
Below we give some of the key definitions we use in our specification. We represent split polynomials as coefficient vectors.

\begin{leancode}
def digits (b : ℕ) (x : ℤ) : ℕ :=
  if x = 0 then 1 else (Nat.log b (Int.natAbs x)) + 1
\end{leancode}
\begin{leancode}
def tk_P (BASE : ℕ) (x y : ℤ) : ℕ :=
  (digits BASE (max x.natAbs y.natAbs))

def tk_i (KX KY BASE : ℕ) (x y : ℤ) : ℕ :=
  (max
    (((digits BASE x) + KX - 1) / KX) -- ceiling-division trick
    (((digits BASE y) + KY - 1) / KY))

def tk_B (KX KY BASE : ℕ) (x y : ℤ) : ℕ :=
  BASE ^ (tk_i KX KY BASE x y)

def tk_split (KX KY KZ BASE : ℕ) (x y : ℕ) (z : ℕ) : Vector ℕ KZ :=
  Vector.ofFn (fun i : Fin KZ =>
    (z / (tk_B KX KY BASE x y) ^ i.1) % (tk_B KX KY BASE x y))
\end{leancode}

We explicitly define the formula for $\theta(k_x, k_y, \vec v)$ in the Lean specification.
We build up to this by defining the formulas for $C_x$, $C_y$, and $C$ per the definitions in Section \ref{sec:derivation}. We then define $\theta$ in terms of these values.
\begin{leancode}
def tk_vmax
  (KX KY : ℕ) (POINTS : Vector ℤ (KX+KY-1)) : ℕ :=
  POINTS.toList.foldl (fun acc z => max acc z.natAbs) 0

def tk_CX
  (BASE KX KY : ℕ)
  (POINTS : Vector ℤ (KX+KY-1))
: ℕ :=
  let v_max := tk_vmax KX KY POINTS
  digits BASE ((∑ i ∈ Finset.range KX, v_max ^ i) : ℤ)

def tk_CY
  (BASE KX KY : ℕ)
  (POINTS : Vector ℤ (KX+KY-1))
: ℕ :=
  let v_max := tk_vmax KX KY POINTS
  digits BASE ((∑ i ∈ Finset.range KY, v_max ^ i) : ℤ)

def tk_C
  (BASE KX KY : ℕ)
  (POINTS : Vector ℤ (KX+KY-1))
: ℕ :=
  (max (tk_CX BASE KX KY POINTS) (tk_CY BASE KX KY POINTS)) + 1

def tk_THETA
  (BASE KX KY : ℕ)
  (POINTS : Vector ℤ (KX+KY-1))
: ℚ :=
  let C := ((tk_C BASE KX KY POINTS) : ℚ)
  let K_MIN := ((min KX KY) : ℚ)
  (C * K_MIN) / (K_MIN - (1 : ℚ))
\end{leancode}

\subsubsection{Algorithm}
Our Lean implementation of the algorithm is a direct rendering of the logic given in Section \ref{sec:algorithm}.
The only significant additions here are as follows: we make it explicit that $\theta(k_x, k_y, \vec v)$ is the base-case threshold; and since we are now dealing with operands in $\Z$, we multiply their absolute values and then determine the sign of the result.

\begin{leancode}
noncomputable def toomk
(BASE KX KY : ℕ)
(h_base : 1 < BASE)
(h_k : 1 < KX ∧ 1 < KY)
(x y : ℤ)
(POINTS : Vector ℤ (KX + KY - 1))
(h_inj : Function.Injective POINTS.get)
: ℚ :=
  if h : (tk_P BASE x y ≤ (tk_THETA BASE KX KY POINTS)) then
    x * y
  else

    let x_abs : ℕ := (x.natAbs)
    let y_abs : ℕ := (y.natAbs)

    let B := tk_B KX KY BASE x_abs y_abs
    let pcoeffs := tk_split KX KY KX BASE x_abs y_abs x_abs
    let qcoeffs := tk_split KX KY KY BASE x_abs y_abs y_abs

    let p_vec := fun i => eval_vec_as_poly pcoeffs (POINTS.get i)
    let q_vec := fun i => eval_vec_as_poly qcoeffs (POINTS.get i)

    let evaluated_pq : Fin (KX + KY - 1) → ℚ :=
      fun i =>
        let a := p_vec i
        let b := q_vec i
        let absprod := toomk BASE KX KY h_base h_k (a.natAbs : ℤ) (b.natAbs : ℤ) POINTS h_inj
        if (Xor (a < 0) (b < 0)) then -absprod else absprod

    let POINTS_Q := fun i => (POINTS.get i : ℚ)
    let r := interpolate KX KY POINTS_Q evaluated_pq
    let absprod := eval_vec_as_poly_Q r B

    if (Xor (x < 0) (y < 0)) then -absprod else absprod
termination_by
  tk_P BASE x y
decreasing_by
  ...
\end{leancode}
We use the notation in Section \ref{sec:algorithm} to explain the implementation here. First, we check if the problem size $P_{x, y}$ is $\leq \theta(k_x, k_y, \vec v)$. If so, we switch to default multiplication. If not, we enter the recursive case.

In the recursive case, we calculate $B$ using {\tt\scriptsize tk\_B}. {\tt\scriptsize x\_abs} and {\tt\scriptsize y\_abs} contain the values of $|x|$ and $|y|$ respectively; it is these values on which we perform the splitting step. {\tt\scriptsize pcoeffs} and {\tt\scriptsize qcoeffs} are coefficient-vector representations of the split polynomials $p$ and $q$ respectively. The vector $\vec w$, which contains the pointwise products of $p(v_j) \cdot q(v_j)$ for each entry $v_j$ of $\vec v$, is denoted by {\tt\scriptsize evaluated\_pq}.

The function {\tt\scriptsize interpolate} performs the Vandermonde interpolation step and returns the result: the interpolated coefficient vector $\vec r$. Evaluating the polynomial represented by $\vec r$ at $B$ gives us the solution, to which we then assign the correct sign.

In keeping with the generality of the specification, we perform the interpolation over $\mathbb{Q}$. The function thus returns a value in $\mathbb{Q}$, and our correctness theorem shows that this value is the same as the integer value $x \cdot y \in \mathbb{Z}$, cast to $\mathbb{Q}$.

\subsection{Termination}
Whenever we specify a recursive function in Lean, Lean will attempt to auto-verify that the function terminates.
If the termination proof is nontrivial, as it is here, then the auto-verification will not suffice and we must construct the proof ourselves.

The task of showing termination depends on a proof that the measure $P$ is decreasing above $\theta$.
We prove the following lemmas with the assistance of Aristotle.
These lemmas include some conversions of $x, y$ to their absolute values. This is because, although the main {\tt\scriptsize toomk} function takes integer operands, it multiplies their absolute values and only assigns the answer its correct sign at the end.
A previous version of the proof used assistance from both Aristotle and GPT-5 to prove these lemmas; as part of the revision, we have re-proved them using Aristotle alone, with GPT's contributions restricted to formalization assistance.

\begin{lemma}\label{l2}
For each entry $v_j$ of $\vec v$, $\delta(p(v_j)) \leq P / k_{\min} + C$.
\end{lemma}
\begin{leancode}
lemma eval_bound_x
  (KX KY BASE : ℕ)
  (h_base : 1 < BASE)
  (hk : 1 < KX ∧ 1 < KY)
  (x y : ℕ)
  (POINTS : Vector ℤ (KX+KY-1))
: ∀ p ∈ POINTS, ((digits BASE (eval_vec_as_poly (tk_split KX KY KX BASE x y x) p) : ℚ))
≤ ((tk_P BASE x y : ℚ) / (min KX KY : ℚ)) + (tk_C BASE KX KY POINTS : ℚ)
\end{leancode}

\begin{lemma} \label{l3}
For each entry $v_j$ of $\vec v$, $\delta(q(v_j)) \leq P / k_{\min} + C$.
\end{lemma}
\begin{leancode}
lemma eval_bound_y
  (KX KY BASE : ℕ)
  (h_base : 1 < BASE)
  (hk : 1 < KX ∧ 1 < KY)
  (x y : ℕ)
  (POINTS : Vector ℤ (KX+KY-1))
: ∀ p ∈ POINTS, ((digits BASE (eval_vec_as_poly (tk_split KX KY KY BASE x y y) p) : ℚ))
≤ ((tk_P BASE x y : ℚ) / (min KX KY : ℚ)) + (tk_C BASE KX KY POINTS : ℚ)
\end{leancode}

We conclude by formalizing and proving Theorem \ref{thm1} in Lean. Because of the upper bound on $P'$ established above, it is easy to show that the problem size strictly decreases above $\theta$.

\begin{lemma} \label{l4}
$P > \theta \implies P/k_{\min} + C < P$.
\end{lemma}
\begin{leancode}
lemma decrease_property_aux
  (KX KY BASE : ℕ)
  (h_base : 1 < BASE)
  (hk : 1 < KX ∧ 1 < KY)
  (x y : ℕ)
  (POINTS : Vector ℤ (KX+KY-1))
  (h_THETA : (tk_P BASE x y) > (tk_THETA BASE KX KY POINTS))
: ((tk_P BASE x y : ℚ) / (min KX KY : ℚ)) + (tk_C BASE KX KY POINTS : ℚ) < (tk_P BASE x y : ℚ)
\end{leancode}

\begin{lemma} \label{l5}
$P > \theta \implies P' < P$.
\end{lemma}
\begin{leancode}
lemma decrease_property
  (KX KY BASE : ℕ)
  (h_base : 1 < BASE)
  (hk : 1 < KX ∧ 1 < KY)
  (x y : ℤ)
  (POINTS : Vector ℤ (KX+KY-1))
  (i : Fin (KX+KY - 1))
  (h_THETA : (tk_P BASE x y) > (tk_THETA BASE KX KY POINTS))
: tk_P BASE
    |(eval_vec_as_poly (tk_split KX KY KX BASE x.natAbs y.natAbs x.natAbs) (POINTS.get i))|
    |(eval_vec_as_poly (tk_split KX KY KY BASE x.natAbs y.natAbs y.natAbs) (POINTS.get i))|
  < tk_P BASE x y
\end{leancode}

Lemmas \ref{l4}-\ref{l5} constitute a restatement of Theorem \ref{thm1}. The underlying idea of Lemmas \ref{l2}-\ref{l3} is that no sub-multiplication will have a problem size in excess of $P / k_{\min} + C$. The problem size of a sub-multiplication is the quantity represented by $P'$. In this manner, we logically progress towards Lemma \ref{l5}.

\subsection{Correctness}
The statement of the correctness theorem is simple: under the classical Toom-Cook assumptions, the return value of the algorithm is the product $x \cdot y$, cast to $\Q$.

\begin{thm}
Given $b, k_x, k_y, x, y, \vec v$ such that $b \in \mathbb{N}$, $1 < b$, $k_x, k_y \in \mathbb{N}, 1 < k_x, 1 < k_y$, $x, y \in \mathbb{Z}$, $\vec v \in \mathbb{Z}^{k_x+k_y-1}$, and $\vec v$ contains distinct points:\\
\textsc{ToomK}$(b, k_x, k_y, x, y, \vec v) = (x \cdot y)$, cast to $\Q$.
\end{thm}
\begin{leancode}
theorem toomk_correctness
(BASE KX KY : ℕ)
(h_base : 1 < BASE)
(h_k : 1 < KX ∧ 1 < KY)
(x y : ℤ)
(POINTS : Vector ℤ (KX + KY - 1))
(h_inj : Function.Injective POINTS.get)
: (toomk BASE KX KY h_base h_k x y POINTS h_inj) = (x * y : ℚ)
\end{leancode}
We provided only the goal header of this theorem to Aristotle. Aristotle was able to reason through the entire proof by itself. It took a fairly straightforward route: a proof of the correctness of the Vandermonde interpolation procedure, followed by strong induction on the problem size.

\section{Conclusion}
We have presented a machine-verified correctness proof of the Toom-Cook multiplication algorithm given arbitrary distinct integer evaluation points. The proof requires a formula for a base-case threshold that ensures the algorithm's termination for any integer values of the operands $x, y$. By analyzing the behavior of the recurrence's subproblem sizes over time, we derive a base case threshold formula $\theta(k_x, k_y, \vec v)$ that suffices to ensure termination. We machine-verify that this $\theta$---which can be expressed entirely in terms of our working base $b$, the point vector $\vec v$, and the splitting parameters $k_x$ and $k_y$---is a valid threshold.
Our proof, which uses the assistance of Aristotle, demonstrates the capability of AI to automate the most time-intensive parts of the machine verification process.

\section{Acknowledgements}
Claude Fable 5.1 and GPT-6 Astra assisted in proofreading the paper's current version. We would also like to thank Evan Lucca, Joseph Donato, and Karun Ram for their ideas and feedback.

\printbibliography

\end{document}